\documentclass[pdflatex,sn-mathphys]{sn-jnl}

\usepackage{natbib}
\usepackage{amsmath}
\usepackage{graphicx}
\usepackage{amssymb}
\usepackage{epstopdf}
\usepackage{amsthm}
\usepackage{bm}
\usepackage{algorithm}
\usepackage{enumitem}
\usepackage{soul,color}
\usepackage{mathtools}
\usepackage{caption}
\usepackage{subcaption}
\usepackage{cases}
\usepackage{tabularx}

\jyear{2022}%

\def\R{{\mathbb R}}

\theoremstyle{thmstyleone}%
\newtheorem{theorem}{Theorem}
\theoremstyle{thmstyletwo}%
\newtheorem{lemma}{Lemma}%

\theoremstyle{thmstylethree}%

\newcommand{\norm}[1]{\|#1\|} 

\newcommand\numberthis{\addtocounter{equation}{1}\tag{\theequation}}

\begin{document}

\title[Preprint]{A Closed-Form Bound on the Asymptotic Linear Convergence of Iterative Methods via Fixed Point Analysis}


\author*[1]{\fnm{Trung} \sur{Vu}}\email{vutru@oregonstate.edu}

\author[1]{\fnm{Raviv} \sur{Raich}}\email{raich@eecs.oregonstate.edu}

\affil*[1]{\orgdiv{School of Electrical Engineering and Computer Science}, \orgname{Oregon State University}, \orgaddress{\city{Corvallis}, \postcode{97331-5501}, \state{Oregon}, \country{USA}}}



\abstract{
In many iterative optimization methods, fixed-point theory enables the analysis of the convergence rate via the contraction factor associated with the linear approximation of the fixed-point operator.
While this factor characterizes the asymptotic linear rate of convergence, it does not explain the non-linear behavior of these algorithms in the non-asymptotic regime.
In this letter, we take into account the effect of the first-order approximation error and present a closed-form bound on the convergence in terms of the number of iterations required for the distance between the iterate and the limit point to reach an arbitrarily small fraction of the initial distance.
Our bound includes two terms: one corresponds to the number of iterations required for the linearized version of the fixed-point operator and the other corresponds to the overhead associated with the approximation error. 
With a focus on the convergence in the scalar case, the tightness of the proposed bound is proven for positively quadratic first-order difference equations.
}

\keywords{non-linear difference equations, asymptotic linear convergence, convergence bounds, fixed-point iterations}



\maketitle

\section{Introduction}

Many iterative optimization methods, such as gradient descent and alternating projections, can be interpreted as fixed-point iterations \cite{polyak1964some,saigal1978efficient,walker2011anderson,jung2017fixed}.
Such methods consist of the construction of a series $\{\bm x^{(k)}\}_{k=0}^\infty \subset \R^n$ generated by
\begin{align} \label{equ:fixed}
    \bm x^{(k+1)} = \mathcal{F}(\bm x^{(k)}) ,
\end{align}
where the fixed-point operator $\mathcal{F}$ is an endomorphism on $\R^n$.
By the fixed-point theorem \cite{brouwer1911abbildung,banach1922operations,lambers2019explorations}, if the Jacobian of $\mathcal{F}$ is bounded uniformly, in the matrix norm $\norm{\cdot}_2$ induced by the Euclidean norm for vectors $\norm{\cdot}$, by $\rho \in (0,1)$, the sequence $\{\bm x^{(k)}\}_{k=0}^\infty$ generated by (\ref{equ:fixed}) converges locally to a fixed-point $\bm x^*$ of $\mathcal{F}$ at a linear rate $\rho$, i.e., $\norm{\bm x^{(k+1)} - \bm x^*} \leq \rho \norm{\bm x^{(k)} - \bm x^*}$ for all integer $k$.\footnote{$\norm{\cdot}$ denotes the Euclidean norm.}
Assume that $\mathcal{F}$ is differentiable at $\bm x^*$ and admits the first-order expansion \cite{roberts1969derivative}
\begin{align*}
    \mathcal{F}(\bm x^{(k)}) = \mathcal{F}(\bm x^*) + \mathcal{T} (\bm x^{(k)} - \bm x^*) + \bm q(\bm x^{(k)} - \bm x^*) ,
\end{align*}
where $\mathcal{T} : \R^n \to \R^n$ is the derivative of $\mathcal{F}$ at $\bm x^*$ and $\bm q : \R^n \to \R^n$ is the residual satisfying $\limsup_{\bm \delta \to 0} \norm{\bm q(\bm \delta)} / \norm{\bm \delta} = 0$.
Then, denoting the error at the $k$-th iteration as $\bm \delta^{(k)} = \bm x^{(k)} - \bm x^*$, the fixed-point iteration (\ref{equ:fixed}) can be viewed as a non-linear but approximately linear difference equation
\begin{align} \label{equ:nonlinear}
    \bm \delta^{(k+1)} = \mathcal{T} (\bm \delta^{(k)}) + \bm q(\bm \delta^{(k)}) .
\end{align}
The stability of non-linear difference equations of form (\ref{equ:nonlinear}) has been studied by Polyak \cite{polyak1964some} in 1964, extending the result from the continuous domain \cite{bellman1953stability}.
In particular, the author showed that if the spectral radius of $\mathcal{T}$, denoted by $\rho(\mathcal{T})$, is strictly less than $1$, then for arbitrarily small $\zeta>0$, there exists a constant $C(\zeta)$ such that $\norm{\bm \delta^{(k)}} \leq C(\zeta) \norm{\bm \delta^{(0)}} (\rho(\mathcal{T}) + \zeta)^k$ with sufficiently small $\norm{\bm \delta^{(0)}}$. While this result characterizes the asymptotic linear convergence of (\ref{equ:nonlinear}), it does not specify the exact conditions on how small $\norm{\bm \delta^{(0)}}$ is as well as how large the factor $C(\zeta)$ is.

This letter develops a more elaborate approach to analyze the convergence of (\ref{equ:nonlinear}) that offers, in addition to the asymptotic linear rate $\rho({\cal T})$, both the region of convergence (i.e., a set ${\cal S}$ such that for any ${\bm \delta}^{(0)} \in {\cal S}$ we have $\lim_{k \to \infty} \norm{{\bm \delta}^{(k)}} = 0$) and a tight closed-form bound on $H(\epsilon)$ defined as the smallest integer guaranteeing $\norm{\bm \delta^{(k)}} \le \epsilon \norm{\bm \delta^{(0)}}$ for $0<\epsilon<1$ and all $k \ge H(\epsilon)$.
We begin with the scalar version of (\ref{equ:nonlinear}) in which the residual term ${q}(\delta)$ is replaced with an exact quadratic function of $\delta$ and then extend the result to the original vector case.
In the first step, we study the convergence of the sequence $\{a_k\}_{k=0}^\infty \subset \mathbb{R}$, generated by the following quadratic first-order difference equation
\begin{align} \label{equ:scalar}
a_{k+1} = \rho a_k + q a_k^2 ,  
\end{align}
where $a_0>0$, $0<\rho<1$, and $q \geq 0$ are real scalars. 
In the second step, we consider the sequence $\{a_k\}_{k=0}^\infty$ obtained by (\ref{equ:scalar}) with $\rho=\rho({\cal T})$ and $q=\sup_{\bm \delta \in \R^n} \frac{\norm{{\bm q}(\bm \delta)}}{\norm{\bm \delta}^2}$ as an upper bound for the sequence $\{\norm{\bm \delta^{(k)}}\}_{k=0}^{\infty}$.
In this letter, we focus on the former step while the latter step is obtained using a more straightforward derivation. 

In analyzing the convergence of $\{a_k\}_{k=0}^\infty$, we focus on tightly characterizing $K(\epsilon)$ (for $0<\epsilon<1$), which is defined as the smallest integer such that  $a_k \leq \epsilon a_0$ for all $k \ge K(\epsilon)$. The value of $K(\epsilon)$ serves as an upper bound on $H(\epsilon)$.
When $q=0$, (\ref{equ:scalar}) becomes a linear first-order difference equation and $\{a_k\}_{k=0}^\infty$ converges uniformly to $0$ at a \textbf{linear rate} $\rho$. In particular, $a_{k+1}=\rho a_k$ implies $a_k=a_0 \rho^k$ for any non-negative integer $k$.  Then, for $q=0$, an exact expression of $K(\epsilon)$ can be obtained in closed-form as
\begin{align} \label{equ:k_uniform}
    K(\epsilon) = \Bigl \lceil \frac{\log(1/\epsilon)}{\log(1/\rho)} \Bigr \rceil .
\end{align}
When $q>0$, the sequence $\{a_k\}_{k=0}^\infty$ either converges, diverges or remains constant depending on the initial value $a_0$:
\begin{enumerate}
\item If $a_0>(1-\rho)/q$, then $\{a_k\}_{k=0}^\infty$ diverges.
\item If $a_0=(1-\rho)/q$, then $a_k=(1-\rho)/q$ for all $k \in {\mathbb N}$.
\item If $a_0<(1-\rho)/q$, then $\{a_k\}_{k=0}^\infty$ converges to $0$ monotonically.
\end{enumerate}
We are interested in the convergence of the sequence $\{a_k\}_{k=0}^\infty$ for $a_0<(1-\rho)/q$. 
In the asymptotic regime ($a_k$ small), the convergence is almost linear since the first-order term $\rho a_k$ dominates the second-order term $q a_k^2$. 
In the early stage ($a_k$ large), on the other hand, the convergence is non-linear due to the strong effect of $q a_k^2$.
In addition, when $\rho \to 0$, one would expect $\{a_k\}_{k=0}^\infty$ enjoys a fast quadratic convergence as $q a_k^2$ dominates $\rho a_k$.
On the other end of the spectrum, when $\rho \to 1$, we observe that the convergence is even slower than linear, making it more challenging to estimate $K(\epsilon)$.

\section{Asymptotic Convergence in the Scalar Case}

\begin{figure}[t]
    \centering
    \begin{subfigure}[b]{0.49\textwidth}
        \centering
        \includegraphics[width=\textwidth]{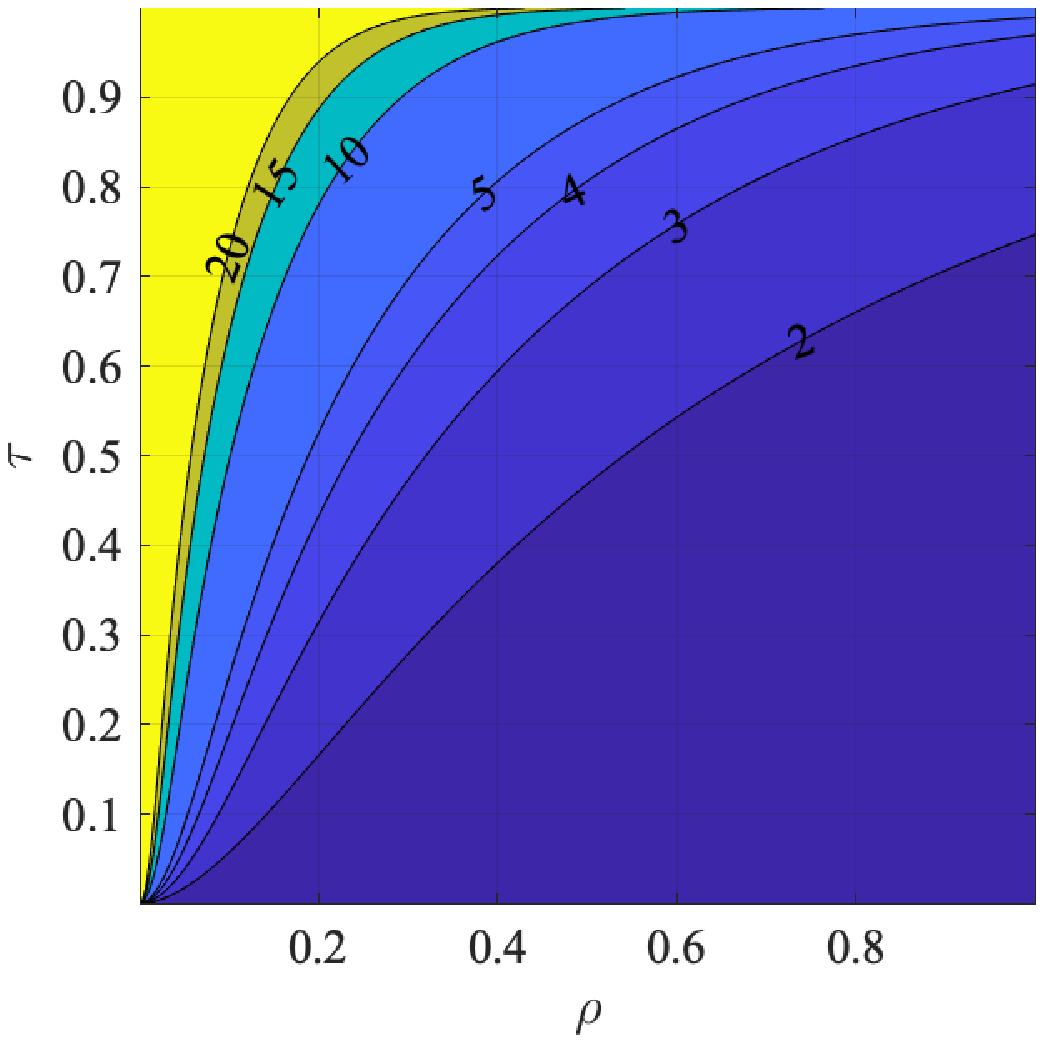}
    \end{subfigure}
    \begin{subfigure}[b]{0.49\textwidth}
        \centering
        \includegraphics[width=\textwidth]{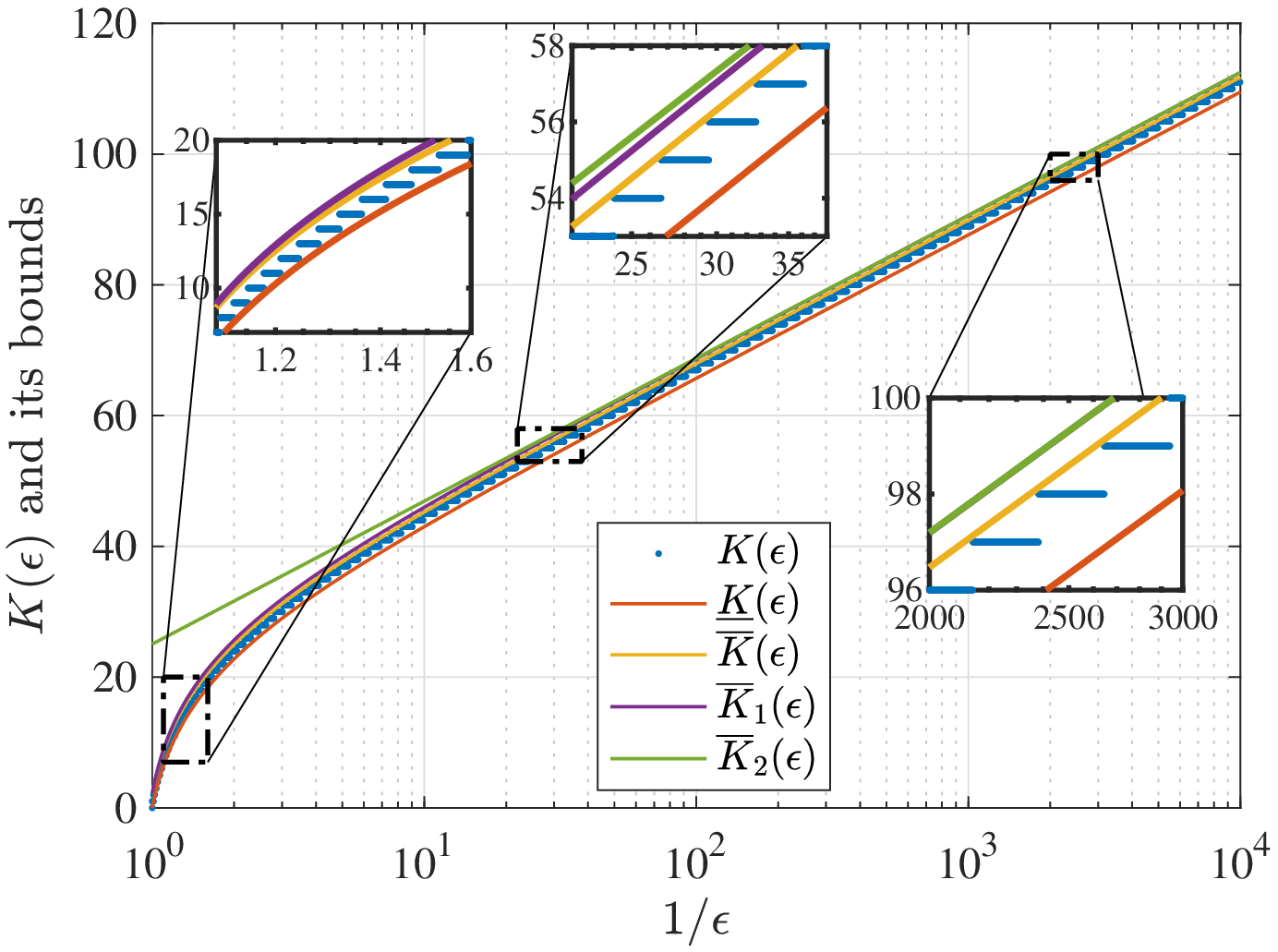}
    \end{subfigure}
    \caption{(Left) Contour plot of the bound on asymptotic gap between $\overline{K}_2(\epsilon)$ and $K(\epsilon)$, given in (\ref{equ:gap}). (Right) Log-scale plot of $K(\epsilon)$ and its bounds as functions of $1/\epsilon$, with $\rho=0.9$ and $\tau=0.89$. Three zoomed plots are added to the original plot for better visualization.} \label{fig:k}
\end{figure}

In this section, we provide a tight upper bound on $K(\epsilon)$ in terms of $a_0$, $\rho$, $q$, and $\epsilon$.
Our bound suggests the sequence $\{a_k\}_{k=0}^\infty$ converges to $0$ at an \textbf{asymptotically linear rate} $\rho$ with an overhead cost that depends on only two quantities: $\rho$ and $a_0 q/(1-\rho)$.
Our main result is stated as follows.
\begin{theorem} \label{theo:scalar}
Consider the sequence $\{a_k\}_{k=0}^\infty$ defined in (\ref{equ:scalar}) with $a_0>0$, $0<\rho<1$, and $q > 0$. Assume that $a_0<(1-\rho)/q$ and denote $\tau = a_0 q / (1-\rho)$ (where $0<\tau<1$).
Then, for any $0<\epsilon<1$, the smallest integer, denoted by $K(\epsilon)$, such that $a_k \leq \epsilon a_0$ for all $k \ge K(\epsilon)$, can be bounded as follows
\begin{align} \label{equ:k1}
    K(\epsilon) \leq \frac{\log(1/\epsilon)}{\log(1/\rho)} + c(\rho,\tau) \triangleq \overline{K}_2(\epsilon) ,
\end{align}
where
\begin{align} \label{equ:c3}
    c(\rho,\tau) = \frac{1}{\rho \log(1/\rho)} \Delta E_1\Bigl(\log\frac{1}{\rho + \tau (1-\rho)} , \log\frac{1}{\rho}\Bigr) +  b(\rho,\tau),
\end{align}
$\Delta E_1(x,y) = E_1(x) - E_1(y)$, $E_1(x) = \int_x^\infty \frac{e^{-t}}{t}dt$ is the exponential integral \cite{milton1964handbook}, and
\begin{align} \label{equ:b}
    b(\rho,\tau) = \frac{1}{2\rho} \log \biggl( \frac{\log(1/\rho)}{\log \bigl( {1}/(\rho+\tau(1-\rho)) \bigr)} \biggr) + 1 . 
\end{align}
Moreover, the gap $\overline{K}_2(\epsilon)-K(\epsilon)$ is upper-bounded asymptotically as follows\footnote{A tighter version of the upper bound $\overline{K}_2(\epsilon)$ is given in Appendix~\ref{sec:proof}, cf., (\ref{equ:gap1}) and (\ref{equ:gap2}).}
\begin{align*}
    \lim_{\epsilon \to 0} \Bigl( \overline{K}_2(\epsilon)-K(\epsilon) \Bigr) &\leq \frac{\Delta E_1\bigl(2\log\frac{1}{\rho+\tau(1-\rho)} , 2\log\frac{1}{\rho}\bigr) - \rho \Delta E_1\bigl(\log\frac{1}{\rho+\tau(1-\rho)} , \log\frac{1}{\rho} \bigr)}{2\rho^2 \log(1/\rho)} \\
    &\qquad + b(\rho,\tau) . \numberthis \label{equ:gap}
\end{align*}
\end{theorem}
\noindent The proof of Theorem~\ref{theo:scalar} is given in Appendix~\ref{sec:proof}.
The upper bound $\overline{K}_2(\epsilon)$, given in (\ref{equ:k1}), is the sum of two terms: (i) the first term is similar to (\ref{equ:k_uniform}), representing the asymptotic linear convergence of $\{a_k\}_{k=0}^\infty$; (ii) the second term, $c(\rho,\tau)$, is independent of $\epsilon$, representing the overhead in the number of iterations caused by the non-linear term $q a_k^2$.
This overhead term is understood as the additional number of iterations beyond the number of iterations for the linear model.
As one would expect, when $a_0 \to (1-\rho)/q$, we have $\tau \to 1$ and $c(\rho,\tau)$ approaches infinity. On the other hand, when $\tau \to 0$, the gap from the number of iterations required by the linear model $c(\rho,\tau)$ approaches $1$. The right hand side (RHS) of (\ref{equ:gap}) is an upper bound on the asymptotic gap between our proposed upper bound on $K(\epsilon)$ and the actual value of $K(\epsilon)$ and hence represents the tightness of our bound. 
The value of the bound as a function of $\rho$ and $\tau$ is shown in Fig.~\ref{fig:k} (left). It is notable that the asymptotic gap is guaranteed to be no more than 10 iterations for a large portion of the $(\rho,\tau)$-space. It is particularly small in the lower right part of the figure. For example, for $\rho \geq 0.9$ and $\tau \leq 0.9$, the gap is no more than 4 iterations.
Figure~\ref{fig:k} (right) demonstrates different bounds on $K(\epsilon)$ (blue dotted line) including $\overline{K}_2(\epsilon)$ (green solid line). We refer the readers to Appendix~\ref{sec:proof} for the details of other bounds in the figure.
 We observe that the upper bound $\overline{K}_2(\epsilon)$ approaches $K(\epsilon)$ as $\epsilon \to 0$, with the asymptotic gap of less than $2$ iterations. On the other hand, $\overline{K}_2(\epsilon)$ reaches $c(\rho,\tau) \approx 25$ as $\epsilon \to 1$, suggesting that the proposed bound $\overline{K}_2(\epsilon)$ requires no more than $25$ iterations beyond the number of iterations required by the linear model to achieve $a_k \leq \epsilon a_0$.

\section{Extension to the Vector Case}

We now consider an extension of Theorem~\ref{theo:scalar} to the convergence analysis in the vector case given by (\ref{equ:nonlinear}). More elaborate applications of the proposed analysis in convergence analysis of iterative optimization methods can be found in \cite{vu2019local,vu2019accelerating,vu2019convergence,vu2021exact}.

\begin{theorem} \label{theo:vector}
Consider the difference equation
\begin{align}
    \bm \delta^{(k+1)} = \bm T \bm \delta^{(k)} + \bm q(\bm \delta^{(k)}) ,
\end{align}
where $\bm T \in \R^{n \times n}$ admits an eigendecomposition $\bm T = \bm Q \bm \Lambda \bm Q^{-1}$, $\bm Q \in \R^{n \times n}$ is an invertible matrix with the condition number $\kappa(\bm Q) = \norm{\bm Q}_2 \norm{\bm Q^{-1}}_2$, and $\bm \Lambda$ is an $n \times n$ diagonal matrix whose entries are strictly less than $1$ in magnitude. In addition, assume that there exists a finite constant $q>0$ satisfying $\norm{\bm q(\bm \delta)} \leq q \norm{\bm \delta}^2$ for any $\bm \delta \in \R^n$. Then, for any $0<\epsilon<1$, we have $\norm{{\bm \delta}^{(k)}} \leq \epsilon \norm{{\bm \delta}^{(0)}}$ provided that
\begin{align} \label{equ:dTk}
    \norm{{\bm \delta}^{(0)}} < \frac{1-\rho(\bm T)}{q \kappa(\bm Q)^2} \text{ and } k \geq \frac{\log(1/\epsilon)+\log(\kappa(\bm Q))}{\log(1/\rho(\bm T))} + c\bigl( \rho(\bm T),\frac{q \kappa(\bm Q) \norm{\bm Q}_2 \norm{\bm Q^{-1} \bm \delta^{(0)}}}{1-\rho(\bm T)} \bigr)
\end{align}
where $c(\rho,\tau)$ is given in (\ref{equ:c3}).
Moreover, if $\bm T$ is symmetric, then (\ref{equ:dTk}) becomes
\begin{align} \label{equ:dTk_sym}
    \norm{{\bm \delta}^{(0)}} < \frac{1-\rho(\bm T)}{q} \text{ and } k \geq \frac{\log(1/\epsilon)}{\log(1/\rho(\bm T))} + c\Bigl( \rho(\bm T),\frac{q \norm{\bm \delta^{(0)}}}{1-\rho(\bm T)} \Bigr) .
\end{align}
\end{theorem}
\noindent Note that the RHS of the inequalities involving $k$ in both (\ref{equ:dTk}) and (\ref{equ:dTk_sym}) serve as upper bounds to $H(\epsilon)$ defined in the introduction. Moreover, the sets of all ${\bm \delta}^{(0)}$ that satisfy the inequality involving ${\bm \delta}^{(0)}$ in both (\ref{equ:dTk}) and (\ref{equ:dTk_sym}) offer valid regions of convergence. Similar to the scalar case, we observe in the number of required iterations one term corresponding to the asymptotic linear convergence and another term corresponding to the non-linear convergence at the early stage. When $\bm T$ is asymmetric, there is an additional cost of diagonalizing $\bm T$, associated with $\kappa(\bm Q)$ in (\ref{equ:dTk}).
The proof of Theorem~\ref{theo:vector} is given in Appendix~\ref{sec:proof2}.

\section{Conclusions}
With a focus on fixed-point iterations, we analyzed the convergence of the sequence generated by a quadratic first-order difference equation. We presented a bound on the minimum number of iterations required for the distance between the iterate and the limit point to reach an arbitrarily small fraction of the initial distance. Our bound includes two terms: one corresponds to the number of iterations required for the linearized difference equation and the other corresponds to the overhead associated with the residual term. The bound for the vector case is derived based on a tight bound obtained for the scalar quadratic difference equation. A characterization of the tightness of the bound for the scalar quadratic difference equation was introduced.



\section*{Statements and Declarations}
The authors have no competing interests to declare that are relevant to the content of this article. Data sharing not applicable to this article as no datasets were generated or analysed during the current study.







\begin{appendices}






\section{Proof of Theorem~\ref{theo:scalar}}
\label{sec:proof}

First, we establish a sandwich inequality on $K(\epsilon)$ in the following lemma:
\begin{lemma} \label{lem:F}
For any $0 < \epsilon < 1$, let $K(\epsilon)$  be the smallest integer such that for all $k \ge K(\epsilon)$, we have $a_k \leq \epsilon a_0$. Then, 
\begin{align} \label{equ:K_epsilon}
    \underline{K}(\epsilon) \triangleq F\bigl(\log(1/\epsilon)\bigr) \leq K(\epsilon) \leq F\bigl(\log(1/\epsilon)\bigr) + b(\rho,\tau) \triangleq \overline{K}(\epsilon) , 
\end{align}
where $b(\rho,\tau)$ is defined in (\ref{equ:b}) and
\begin{align} \label{def:f}
   F(x) = \int_0^x f(t) dt  \quad \text{ with } \quad  f(x) &= \frac{1}{-\log \bigl(\rho + \tau (1-\rho) e^{-x} \bigr)} .
\end{align}
\end{lemma}
\noindent The lemma provides an upper bound on $K(\epsilon)$. Moreover, it is a tight bound in the sense that the gap between lower bound $\underline{K}(\epsilon)$ and the upper bound $\overline{K}(\epsilon)$ is independent of $\epsilon$. 
In other words, the ratio $K(\epsilon)/\overline{K}(\epsilon)$ approaches $1$ as $\epsilon \to 0$.
Next, we proceed to obtain a tight closed-form upper bound on $\overline{K}(\epsilon)$ by upper-bounding $F(\log(1/\epsilon))$.
\begin{lemma} \label{lem:F_ep}
Consider the function $F(\cdot)$ given in (\ref{def:f}). For $0<\epsilon<1$, we have
\begin{align*}
    F\bigl(\log(1/\epsilon)\bigr) &\leq \frac{\log(1/\epsilon)}{\log(1/\rho)} + \frac{\Delta E_1 \Bigl(\log\frac{1}{\rho+\tau(1-\rho)} , \log\frac{1}{\rho+\epsilon\tau(1-\rho)}\Bigr)}{\rho \log(1/\rho)} \triangleq \overline{F}_1\bigl(\log(1/\epsilon)\bigr) \numberthis \label{equ:F_ep1} \\
    &\leq \frac{\log(1/\epsilon)}{\log(1/\rho)} + \frac{\Delta E_1 \Bigl(\log\frac{1}{\rho+\tau(1-\rho)} , \log\frac{1}{\rho}\Bigr)}{\rho \log(1/\rho)} \triangleq \overline{F}_2 \bigl(\log(1/\epsilon)\bigr) \numberthis \label{equ:F_ep}
\end{align*}
and 
\begin{align*} 
    F\bigl(\log(1/\epsilon)\bigr) &\geq \overline{F}_1\bigl(\log(1/\epsilon)\bigr) - A(\epsilon) \triangleq \underline{F}_1 \bigl(\log(1/\epsilon)\bigr) , \numberthis \label{equ:F_ep1lo}
\end{align*}
where
\begin{align*} 
    A(\epsilon) &\triangleq \frac{\Delta E_1\bigl(2\log\frac{1}{\rho+\tau(1-\rho)} , 2\log\frac{1}{\rho+\tau(1-\rho)\epsilon}\bigr) - \rho \Delta E_1\bigl(\log\frac{1}{\rho+\tau(1-\rho)} , \log\frac{1}{\rho+\tau(1-\rho)\epsilon}\bigr)}{2\rho^2 \log(1/\rho)} . \numberthis \label{equ:A}
\end{align*}
\end{lemma}
\noindent Lemma~\ref{lem:F_ep} offers two upper bounds on $F(\log(1/\epsilon))$ and one lower bound. The first bound $\overline{F}_1(\log(1/\epsilon))$ approximates well the behavior of $F(\log(1/\epsilon))$ for both small and large values of $\log(1/\epsilon)$. The second bound $\overline{F}_2 (\log(1/\epsilon))$ provides a linear bound on $F(\log(1/\epsilon))$ in terms of $\log(1/\epsilon)$. Moreover, the gap between $F(\log(1/\epsilon))$ and $\underline{F}_1(\log(1/\epsilon))$, given by $A(\epsilon)$, can be upper bound by $A(0)$ since $A(\cdot)$ is monotonically decreasing for $\epsilon \in [0,1)$. While $F(\cdot)$ asymptotically increases like $\log(1/\epsilon)/ \log(1/\rho)$, the gap approaches a constant independent of $\epsilon$.
Replacing $F(\log(1/\epsilon))$ on the RHS of (\ref{equ:K_epsilon}) by either of the upper bounds in Lemma~\ref{lem:F_ep}, we obtain two corresponding bounds on $K(\epsilon)$:
\begin{align} \label{equ:K12u}
    \overline{K}_1(\epsilon) \triangleq \overline{F}_1\bigl(\log(1/\epsilon)\bigr) + b(\rho,\tau) \leq \overline{F}_2\bigl(\log(1/\epsilon)\bigr) + b(\rho,\tau) \triangleq \overline{K}_2(\epsilon) ,
\end{align}
where we note that $\overline{K}_2(\epsilon)$ has the same expression as in (\ref{equ:k1}). 
Moreover, the tightness of these two upper bounds can be shown as follows. 
First, using the first inequality in (\ref{equ:K_epsilon}) and then the lower bound on $F(\log(1/\epsilon))$ in (\ref{equ:F_ep1lo}), the gap between $\overline{K}_1(\epsilon)$ and $K(\epsilon)$ can be bounded by
\begin{align*}
    \overline{K}_1(\epsilon) - K(\epsilon) &\leq \overline{K}_1(\epsilon) - F\bigl(\log(1/\epsilon)\bigr) \\
    &\leq \overline{K}_1(\epsilon) - \Bigl( \overline{F}_1\bigl(\log(1/\epsilon)\bigr) - A(\epsilon) \Bigr) \\
    &= \Bigl( \overline{F}_1 \bigl( \bigl(\log(1/\epsilon)\bigr) + b(\rho,\tau) \Bigr) - \Bigl( \overline{F}_1\bigl(\log(1/\epsilon)\bigr) - A(\epsilon) \Bigr) \\
    &= A(\epsilon) + b(\rho,\tau) \\
    &\leq A(0) + b(\rho,\tau) , \numberthis \label{equ:gap1}
\end{align*}
where the last inequality stems from the monotonicity of $A(\cdot)$ in $[0,1)$.
Note that the bound in (\ref{equ:gap1}) holds uniformly independent of $\epsilon$, implying $\overline{K}_1(\epsilon)$ is a tight bound on $K(\epsilon)$. Second, using (\ref{equ:K12u}), the gap between $\overline{K}_2(\epsilon)$ and $K(\epsilon)$ can be represented as
\begin{align*}
    \overline{K}_2(\epsilon) - K(\epsilon) &= \bigl( \overline{K}_2(\epsilon) - \overline{K}_1(\epsilon) \bigr) + \bigl( \overline{K}_1(\epsilon) - K(\epsilon) \bigr) \\
    &= \bigl( \overline{F}_2\bigl(\log(1/\epsilon)\bigr) - \overline{F}_1\bigl(\log(1/\epsilon)\bigr) \bigr) + \bigl( \overline{K}_1(\epsilon) - K(\epsilon) \bigr) \\
    &\leq \bigl( \overline{F}_2\bigl(\log(1/\epsilon)\bigr) - \overline{F}_1\bigl(\log(1/\epsilon)\bigr) \bigr) + \bigl( A(0) + b(\rho,\tau) \bigr) , \numberthis \label{equ:gap20}
\end{align*}
where the last inequality stems from (\ref{equ:gap1}). Furthermore, using the definition of $\overline{F}_1(\log(1/\epsilon))$ and $\overline{F}_2(\log(1/\epsilon))$ in (\ref{equ:F_ep1}) and (\ref{equ:F_ep}), respectively, we have $\lim_{\epsilon \to 0} (\overline{F}_2(\log(1/\epsilon)) - \overline{F}_1(\log(1/\epsilon))) = 0$. Thus, taking the limit $\epsilon \to 0$ on both sides of (\ref{equ:gap20}), we obtain
\begin{align*}
    \lim_{\epsilon \to 0} \bigl( \overline{K}_2(\epsilon) - K(\epsilon) \bigr) &\leq A(0) + b(\rho,\tau) . \numberthis \label{equ:gap2}
\end{align*}
We note that $\overline{K}_2(\epsilon)$ is a simple bound that is linear in terms of $\log(1/\epsilon)$ and approaches the upper bound $\overline{K}_1(\epsilon)$ in the asymptotic regime ($\epsilon \to 0$).
Evaluating $A(0)$ from (\ref{equ:A}) and substituting it back into (\ref{equ:gap2}) yields (\ref{equ:gap}), which completes our proof of Theorem~\ref{theo:scalar}.
Figure~\ref{fig:k} (right) depicts the aforementioned bounds on $K(\epsilon)$.
It can be seen from the plot that all the four bounds match the asymptotic rate of increment in $K(\epsilon)$ (for large values of $1/\epsilon$).
The three bounds $\underline{K}(\epsilon)$ (red), $\overline{K}(\epsilon)$ (yellow), and $\overline{K}_1(\epsilon)$ (purple) closely follow $K(\epsilon)$ (blue), indicating that the integral function $F(\cdot)$ effectively estimates the minimum number of iterations required to achieve $a_k \leq \epsilon a_0$ in this setting.
The upper bound $\overline{K}_2(\epsilon)$ (green) forms a tangent to $\overline{K}_1(\epsilon)$ at $1/\epsilon \to \infty$ (i.e., $\epsilon \to 0$).

\subsection{Proof of Lemma~\ref{lem:F}}

\noindent Let $d_k=\log({a_0}/{a_k})$  for each $k \in {\mathbb N}$.  Substituting $a_k = a_0 e^{-d_k}$ into (\ref{equ:scalar}), we obtain the surrogate sequence $\{d_k\}_{k=0}^\infty$:
\begin{align} \label{equ:dn}
    d_{k+1} = d_k - \log \bigl(\rho + \tau (1-\rho) e^{-d_k} \bigr) , 
\end{align}
where $d_0=0$ and $\tau = a_0 q/(1-\rho) \in (0,1)$.
Since $\{a_k\}_{k=0}^\infty$ is monotonically decreasing to $0$ and $d_k$ is monotonically decreasing as a function of $a_k$, $\{d_k\}_{k=0}^\infty$ is a monotonically increasing sequence. 
Our key steps in this proof are first to tightly bound the index $K \in \mathbb{N}$ using $F(d_K)$
\begin{align} \label{equ:d_K}
    F(d_K) \leq K \leq F(d_K) + \frac{1}{2\rho} \log \biggl( \frac{\log\rho}{\log\bigl(\rho+\tau(1-\rho)\bigr)} \biggr) 
\end{align}
and then to obtain (\ref{equ:K_epsilon}) from (\ref{equ:d_K}) using the monotonicity of the sequence $\{d_k\}_{k=0}^\infty$ and of the function $F(\cdot)$. We proceed with the details of each of the steps in the following.

\vspace{5pt}
\noindent \textbf{Step 1:} 
We prove (\ref{equ:d_K}) by showing the lower bound on $K$ first and then showing the upper bound on $K$. Using (\ref{def:f}), we can rewrite (\ref{equ:dn}) as $d_{k+1}=d_k+1/f(d_k)$. Rearranging this equation yields
\begin{align} \label{equ:fdn}
    f(d_k) (d_{k+1} - d_k) = 1 .
\end{align}
Since $f(x)$ is monotonically decreasing, we obtain the lower bound on $K$ in (\ref{equ:d_K}) by
\begin{align*}
    F(d_K) &= \int_0^{d_K} f(x) dx = \sum_{k=0}^{K-1} \int_{d_k}^{d_{k+1}} f(x) dx \\
    &\leq \sum_{k=0}^{K-1} \int_{d_k}^{d_{k+1}} f(d_k) dx = \sum_{k=0}^{K-1} f(d_k) (d_{k+1}-d_k) = K , \numberthis \label{equ:lower_bound}
\end{align*}
where the last equality stems from (\ref{equ:fdn}). For the upper bound on $K$ in (\ref{equ:d_K}), we use the convexity of $f(\cdot)$ to lower-bound $F(d_K)$ as follows
\begin{align*}
    F(d_K) &= \sum_{k=0}^{K-1} \int_{d_k}^{d_{k+1}} f(x) dx \geq \sum_{k=0}^{K-1} \int_{d_k}^{d_{k+1}} \bigl( f(d_k) + f'(d_k) (x-d_k) \bigr) dx \\ &= \sum_{k=0}^{K-1} \Bigl( f(d_k) (d_{k+1}-d_k) + \frac{1}{2} f'(d_k) (d_{k+1}-d_k)^2 \Bigr) . \numberthis \label{equ:F_dK}
\end{align*}
Using (\ref{equ:fdn}) and substituting $f'(x) = -\bigl(f(x)\bigr)^2 \frac{\tau (1-\rho) e^{-x}}{\rho+\tau (1-\rho) e^{-x}}$ into the RHS of (\ref{equ:F_dK}), we obtain
\begin{align} \label{equ:upper_bound0}
    F(d_K) &\geq K - \frac{1}{2} \sum_{k=0}^{K-1} \frac{\tau (1-\rho) e^{-d_k}}{\rho+\tau (1-\rho) e^{-d_k}} .
\end{align}
Note that (\ref{equ:upper_bound0}) already offers an upper on $K$ in terms of $F(d_K)$. To obtain the upper bound on $K$ in (\ref{equ:d_K}) from (\ref{equ:upper_bound0}), it suffices to show that 
\begin{align} \label{equ:sum_frac_dk}
    \sum_{k=0}^{K-1} \frac{\tau (1-\rho) e^{-d_k}}{\rho+\tau (1-\rho) e^{-d_k}} &\leq \frac{1}{\rho} \log \biggl( \frac{\log\rho}{\log\bigl(\rho+\tau(1-\rho)\bigr)} \biggr) . 
\end{align}
In the following, we prove (\ref{equ:sum_frac_dk}) by introducing the functions
\begin{align} \label{def:g}
    g(x) = \frac{\tau (1-\rho) e^{-x}}{\rho+\tau (1-\rho) e^{-x}} \frac{1}{-\log \bigl(\rho + \tau (1-\rho) e^{-x} \bigr)}
\end{align}
and
\begin{align} \label{def:G}
    G(x) = \int_0^{x} g(t) dt =\log\Bigl(\frac{\log(\rho+\tau(1-\rho)e^{-x})}{\log(\rho+\tau(1-\rho))}\Bigr) .
\end{align}
Note that $g(\cdot)$ is monotonically decreasing (a product of two decreasing functions) while $G(\cdot)$ is monotonically increasing (an integral of a non-negative function) on $[0,\infty)$.
We have
\begin{align*}
    G(d_K) &= \int_0^{d_K} g(x) dx = \sum_{k=0}^{K-1} \int_{d_k}^{d_{k+1}} g(x) dx \geq \sum_{k=0}^{K-1} \int_{d_k}^{d_{k+1}} g(d_{k+1}) dx \\
    &= \sum_{k=0}^{K-1} g(d_{k+1}) (d_{k+1}-d_k) = \sum_{k=0}^{K-1} \frac{g(d_{k+1})}{g(d_k)} g(d_k) (d_{k+1}-d_k) . \numberthis \label{equ:GdN}
\end{align*}  

\begin{lemma} \label{lem:fgdn}
For any $k \in \mathbb{N}$, we have ${g(d_{k+1})}/{g(d_k)} \geq \rho$.
\end{lemma}

\begin{proof}
For $k\in \mathbb{N}$, let $t_k = \rho + \tau (1-\rho) e^{-d_k} \in (\rho,1)$. From (\ref{equ:dn}), we have $t_k = e^{-(d_{k+1}-d_k)}$ and $t_{k+1} = \rho + \tau (1-\rho) e^{-d_{k+1}} = \rho + \tau (1-\rho) e^{-d_{k}} e^{-(d_{k+1}-d_k)} = \rho + (t_k-\rho)t_k$.
Substituting $d_k$ for $x$ in $g(x)$ from (\ref{def:g}) and replacing $\rho + \tau (1-\rho) e^{-d_k}$ with $t_k$ yield $g(d_k) = \frac{\tau (1-\rho) e^{-d_{k}}}{t_k} \frac{1}{-\log(t_k)}$.
Repeating the same process to obtain $g(d_{k+1})$ and taking the ratio between $g(d_{k+1})$ and $g(d_k)$, we obtain
\begin{align} \label{equ:g_ratio}
    \frac{g(d_{k+1})}{g(d_k)} = e^{-(d_{k+1}-d_k)} \frac{t_k}{t_{k+1}} \frac{\log(t_k)}{\log(t_{k+1})} .
\end{align}
Substituting $e^{-(d_{k+1}-d_k)} = t_k$ and $t_{k+1} = \rho + (t_k-\rho)t_k$ into (\ref{equ:g_ratio}) yields
\begin{align} \label{equ:frac_g}
    \frac{g(d_{k+1})}{g(d_k)} = \frac{t_k^2 \log(t_k)}{(\rho+(t_k-\rho)t_k)\log(\rho+(t_k-\rho)t_k)} .
\end{align}
We now continue to bound the ratio ${g(d_{k+1})}/{g(d_k)}$ by bounding the RHS. Since $t_k-\rho \ge 0$ and $t_k<1$, we have $t_k-\rho > (t_k-\rho) t_k$ and hence ${t_k}/{(\rho+(t_k-\rho)t_k)} > 1$.
Thus, in order to prove $\frac{g(d_{k+1})}{g(d_k)} \geq \rho$ from the fact that the RHS of (\ref{equ:frac_g}) is greater or equal to $\rho$, it remains to show that 
\begin{align} \label{equ:tp}
    \frac{t_k \log(t_k)}{\log \bigl(\rho+(t_k-\rho)t_k \bigr)} \geq \rho  .
\end{align}
By the concavity of $\log(\cdot)$, it holds that $\log (\frac{\rho}{t_k} 1+\frac{t_k-\rho}{t_k} t ) \geq \frac{\rho}{t_k} \log(1) + \frac{t_k-\rho}{t_k} \log(t_k) = (1-\frac{\rho}{t_k}) \log(t_k)$.
Adding $\log(t_k)$ to both sides of the last inequality yields $\log (\rho+(t_k-\rho)t_k ) \geq (2-\frac{\rho}{t_k}) \log(t_k)$.
Now using the fact that $(\sqrt{\rho/t_k} - \sqrt{t_k/\rho})^2 \geq 0$, we have $2-\rho/t_k \leq t_k/\rho$.
By this inequality and the negativity of $\log(t_k)$, we have $\log (\rho+(t_k-\rho)t_k ) \geq \frac{t_k}{\rho} \log(t_k)$. 
Multiplying both sides by the negative ratio $\rho/\log(\rho+(t_k-\rho)t_k)$ and adjusting the direction of the inequality
yields the inequality in (\ref{equ:tp}), which completes our proof of the lemma. 
\end{proof}

Back to our proof of Theorem~\ref{theo:scalar}, applying Lemma~\ref{lem:fgdn} to (\ref{equ:GdN}) and substituting $d_{k+1}-d_k = -\log(\rho + \tau (1-\rho) e^{-d_k})$ from (\ref{equ:dn}) and $g(d_k)$ from (\ref{def:g}), we have
\begin{align} \label{equ:G_lo}
     G(d_K) &\geq \sum_{k=0}^{K-1} \rho  g(d_k) (d_{k+1}-d_k) = \rho \sum_{k=0}^{K-1} \frac{\tau (1-\rho) e^{-d_k}}{\rho+\tau (1-\rho) e^{-d_k}} .
\end{align}
Using the monotonicity of $G(\cdot)$, we upper-bound $G(d_K)$ by 
\begin{align} \label{equ:G_up}
    G(d_K) \leq G(\infty) = \log \Bigl( \frac{\log \rho}{\log(\rho+\tau(1-\rho))} \Bigr) .
\end{align}
Thus, the RHS of (\ref{equ:G_lo}) is upper bounded by the RHS of (\ref{equ:G_up}). Dividing the result by $\rho$, we obtain (\ref{equ:sum_frac_dk}).
This completes our proof of the upper bound on $K$ in (\ref{equ:d_K}) and thereby the first step of the proof.

\vspace{5pt}
\noindent \textbf{Step 2:}
We proved both the lower bound and the upper bound on $K$ in (\ref{equ:d_K}).
Next, we proceed to show (\ref{equ:K_epsilon}) using (\ref{equ:d_K}).
By the definition of $K(\epsilon)$, $a_{K(\epsilon)} \leq \epsilon a_0 < a_{K(\epsilon)-1}$. Since $d_k=\log({a_0}/{a_k})$, for $k \in {\mathbb N}$, we have $d_{K(\epsilon)-1} \leq \log(1/\epsilon) \leq d_{K(\epsilon)}$.
On the one hand, using the monotonicity of $F(\cdot)$ and substituting $K=K(\epsilon)$ into the lower bound on $K$ in (\ref{equ:d_K}) yields
\begin{align} \label{equ:N_lo}
    F\bigl(\log(1/\epsilon)\bigr) \leq F(d_{K(\epsilon)}) \leq K(\epsilon) .
\end{align}
On the other hand, substituting $K=K(\epsilon)-1$ into the upper bound on $K$ in (\ref{equ:d_K}), we obtain
\begin{align} \label{equ:KFd1}
    K(\epsilon) - 1 \leq F(d_{K(\epsilon)-1}) + \frac{1}{2\rho}\log \biggl( \frac{\log\rho}{\log \bigl(\rho+\tau(1-\rho)\bigr)} \biggr) .
\end{align}
Since $F(\cdot)$ is monotonically increasing and $d_{K(\epsilon)-1} \leq \log(1/\epsilon)$, we have $F(d_{K(\epsilon)-1}) \leq F(\log(1/\epsilon))$.
Therefore, upper-bounding $F(d_{K(\epsilon)-1})$ on the RHS of (\ref{equ:KFd1}) by $F(\log(1/\epsilon))$ yields
\begin{align} \label{equ:N_up}
    K(\epsilon) \leq F\bigl(\log(1/\epsilon)\bigr) + \frac{1}{2\rho}\log \biggl( \frac{\log\rho}{\log \bigl(\rho+\tau(1-\rho)\bigr)} \biggr) + 1 .
\end{align} 
The inequality (\ref{equ:K_epsilon}) follows on combining (\ref{equ:N_lo}) and (\ref{equ:N_up}).

\subsection{Proof of Lemma~\ref{lem:F_ep}}

Let $\nu=\tau (1-\rho) /\rho$. We represent $f(x)$ in the interval $(0,\log(1/\epsilon))$ as
\begin{align*}
    f(x) &= \frac{1}{-\log \bigl(\rho + \tau (1-\rho) e^{-x} \bigr)} \\
    &= \frac{1}{\log(1/\rho)} + \frac{1}{\log(1/\rho)} \frac{\log (1+\nu e^{-x})}{\log(1/\rho)-\log (1+\nu e^{-x})} .
\end{align*}
Then, taking the integral from $0$ to $\log(1/\epsilon)$ yields
\begin{align*}
    F\bigl(\log(1/\epsilon)\bigr) = \frac{1}{\log(1/\rho)} \biggl( \log(1/\epsilon) + \int_0^{\log(1/\epsilon)} \frac{\log (1+\nu e^{-t})}{\log(1/\rho)-\log (1+\nu e^{-t})} dt \biggr) . \numberthis \label{equ:Fx}
\end{align*}
Using $\alpha(1-\alpha/2)=\alpha-\alpha^2/2 \leq \log(1+\alpha) \leq \alpha$, for $\alpha =\nu e^{-t} \geq 0$, on the numerator within the integral in (\ref{equ:Fx}) and changing the integration variable $t$ to $z = \log(1/\rho)-\log (1+ \nu e^{-t})$, we obtain both an upper bound and a lower bound on the integral on the RHS of (\ref{equ:Fx})
\begin{align*}
    \frac{1}{\rho} \int^{\overline{z}}_{\underline{z}} \frac{e^{-z} - \frac{1}{2\rho}e^{-z}(e^{-z}-\rho)}{z} dz  &\leq \int_0^{\log(1/\epsilon)} \frac{\log (1+\nu e^{-t})}{\log(1/\rho)-\log (1+\nu e^{-t})} dt \\
    &\leq \frac{1}{\rho} \int^{\overline{z}}_{\underline{z}} \frac{e^{-z}}{z} dz , \numberthis \label{equ:exp_int}
\end{align*}
where $\underline{z} = -\log(\rho+\tau(1-\rho))$ and $\overline{z} = -\log(\rho+\epsilon\tau(1-\rho))$.
Replacing the integral in (\ref{equ:Fx}) by the upper bound and lower bound from (\ref{equ:exp_int}), using the definition of the exponential integral, and simplifying, we obtain the upper-bound on $F(\log(1/\epsilon))$ given by $\overline{F}_1(\log(1/\epsilon))$ in (\ref{equ:F_ep1}) and similarly the lower bound on $F(\log(1/\epsilon))$ given by $\underline{F}_1(\log(1/\epsilon))$ in (\ref{equ:F_ep1lo}).
Finally, we prove the second upper bound in (\ref{equ:F_ep}) as follows. Since $E_1(\cdot)$ is monotonically decreasing and $\frac{1}{\rho+\epsilon\tau(1-\rho)} \le \frac{1}{\rho}$, we have $E_1(\log\frac{1}{\rho+\epsilon\tau(1-\rho)}) \ge E_1(\log\frac{1}{\rho})$, which implies $\Delta E_1 (\log\frac{1}{\rho+\tau(1-\rho)} , \log\frac{1}{\rho+\epsilon\tau(1-\rho)}) \leq \Delta E_1 (\log\frac{1}{\rho+\tau(1-\rho)} , \log\frac{1}{\rho})$. Combining this with the definition of $\overline{F}_1(\log(1/\epsilon))$ and $\overline{F}_2(\log(1/\epsilon))$ in (\ref{equ:F_ep1}) and (\ref{equ:F_ep}), respectively, we conclude that $\overline{F}_1(\log(1/\epsilon)) \leq \overline{F}_2(\log(1/\epsilon))$, thereby completes the proof of the lemma.

\section{Proof of Theorem~\ref{theo:vector}}
\label{sec:proof2}

Let $\tilde{\bm \delta}^{(k)} = \bm Q^{-1} \bm \delta^{(k)}$ be the transformed error vector. Substituting $\mathcal{T}(\bm \delta^{(k)}) = \bm Q \bm \Lambda \bm Q^{-1} (\bm \delta^{(k)})$ into (\ref{equ:nonlinear}) and then left-multiplying both sides by $\bm Q^{-1}$, we obtain
\begin{align} \label{equ:delta_tilde}
    \tilde{\bm \delta}^{(k+1)} = \bm \Lambda \tilde{\bm \delta}^{(k)} + \tilde{\bm q}(\tilde{\bm \delta}^{(k)}) ,
\end{align}
where $\tilde{\bm q}(\tilde{\bm \delta}^{(k)}) = \bm Q^{-1} \bm q(\bm Q \tilde{\bm \delta}^{(k)})$ satisfies $\norm{\tilde{\bm q}(\tilde{\bm \delta}^{(k)})} \leq q \norm{\bm Q^{-1}}_2 \norm{\bm Q}_2^2 \norm{\tilde{\bm \delta}^{(k)}}^2$.
Taking the norm of both sides of (\ref{equ:delta_tilde}) and using the triangle inequality yield
\begin{align*}
    \norm{\tilde{\bm \delta}^{(k+1)}} &\leq \norm{\bm \Lambda \tilde{\bm \delta}^{(k)}} + \norm{\tilde{\bm q}(\tilde{\bm \delta}^{(k)})} \\
    &\leq \norm{\bm \Lambda}_2 \norm{\tilde{\bm \delta}^{(k)}} + q \norm{\bm Q^{-1}}_2 \norm{\bm Q}_2^2 \norm{\tilde{\bm \delta}^{(k)}}^2
\end{align*}
Since $\norm{\bm \Lambda}_2 = \rho(\mathcal{T})$, the last inequality can be rewritten compactly as
\begin{align}
    \norm{\tilde{\bm \delta}^{(k+1)}} \leq \rho \norm{\tilde{\bm \delta}^{(k)}} + \tilde{q} \norm{\tilde{\bm \delta}^{(k)}}^2 ,
\end{align}
where $\rho=\rho(\mathcal{T})$ and $\tilde{q} = q \norm{\bm Q^{-1}}_2 \norm{\bm Q}_2^2$.

To analyze the convergence of $\{\norm{\tilde{\bm \delta}^{(k)}}\}_{k=0}^\infty$, let us consider a surrogate sequence $\{a_k\}_{k=0}^\infty \subset \R$ defined by $a_{k+1} = \rho a_k + \tilde{q} a_k^2$ with $a_0=\norm{\tilde{\bm \delta}^{(0)}}$.
We show that $\{a_k\}_{k=0}^\infty$ upper-bounds $\{\norm{\tilde{\bm \delta}^{(k)}}\}_{k=0}^\infty$, i.e.,
\begin{align} \label{equ:ak}
    \norm{\tilde{\bm \delta}^{(k)}} \leq a_k \quad \forall k \in \mathbb{N}  .
\end{align}
The base case when $k=0$ holds trivially as $a_0=\norm{\tilde{\bm \delta}^{(0)}}$.
In the induction step, given $\norm{\tilde{\bm \delta}^{(k)}} \leq a_k$ for some integer $k \geq 0$, we have
\begin{align*}
    \norm{\tilde{\bm \delta}^{(k+1)}} &\leq \rho \norm{\tilde{\bm \delta}^{(k)}} + \tilde{q} \norm{\tilde{\bm \delta}^{(k)}}^2 \leq \rho a_k + \tilde{q} a_k^2 = a_{k+1} .
\end{align*}
By the principle of induction, (\ref{equ:ak}) holds for all $k \in \mathbb{N}$.
Assume for now that $a_0 = \norm{\tilde{\bm \delta}^{(0)}} < (1-\rho)/\tilde{q}$, then applying Theorem~\ref{theo:scalar} yields $a_k \leq \tilde{\epsilon} a_0$ for any $\tilde{\epsilon}>0$ and integer $k \geq {\log(1/\tilde{\epsilon})}/{\log(1/\rho)} + c(\rho,\tau)$.
Using (\ref{equ:ak}) and setting $\tilde{\epsilon} = \epsilon/\kappa(\bm Q)$, we further have $\norm{\tilde{\bm \delta}^{(k)}} \leq a_k \leq \tilde{\epsilon} a_0 = \epsilon \norm{\tilde{\bm \delta}^{(0)}} / \kappa(\bm Q)$ for all
\begin{align} \label{equ:kk}
    k \geq \frac{\log(1/\epsilon)+\log(\kappa(\bm Q))}{\log(1/\rho)} + c\Bigl( \rho,\frac{\tilde{q} \norm{\tilde{\bm \delta}^{(0)}}}{1-\rho} \Bigr) .
\end{align}
Now, it remains to prove \textit{(i)} the accuracy on the transformed error vector $\norm{\tilde{\bm \delta}^{(k)}} \leq \tilde{\epsilon} \norm{\tilde{\bm \delta}^{(0)}}$ is sufficient for the accuracy on the original error vector $\norm{{\bm \delta}^{(k)}} \leq \epsilon \norm{{\bm \delta}^{(0)}}$; and \textit{(ii)} the initial condition $\norm{{\bm \delta}^{(0)}} < (1-\rho)/(q \kappa(\bm Q)^2)$ is sufficient for $\norm{\tilde{\bm \delta}^{(0)}} < (1-\rho)/\tilde{q}$.
In order to prove \textit{(i)}, using $\norm{\tilde{\bm \delta}^{(k)}} \leq \epsilon \norm{\tilde{\bm \delta}^{(0)}}/\kappa(\bm Q)$, we have
\begin{align*}
    \norm{\bm \delta^{(k)}} = \norm{\bm Q \tilde{\bm \delta}^{(k)}} &\leq \norm{\bm Q}_2 \norm{\tilde{\bm \delta}^{(k)}} \\
    &\leq \norm{\bm Q}_2 \frac{\epsilon}{\norm{\bm Q}_2 \norm{\bm Q^{-1}}_2} \norm{\tilde{\bm \delta}^{(0)}} = \frac{\epsilon}{\norm{\bm Q^{-1}}_2} \norm{\tilde{\bm \delta}^{(0)}} \leq \epsilon \norm{\bm \delta^{(0)}} ,
\end{align*}
where the last inequality stems from $\norm{\tilde{\bm \delta}^{(0)}} = \norm{\bm Q^{-1} {\bm \delta}^{(0)}} \leq \norm{\bm Q^{-1}}_2 \norm{\bm \delta^{(0)}}$.
To prove \textit{(ii)}, we use similar derivation as follows
\begin{align*}
    \norm{\tilde{\bm \delta}^{(0)}} \leq \norm{\bm Q^{-1}}_2 \norm{\bm \delta^{(0)}} < \norm{\bm Q^{-1}}_2  \frac{1-\rho}{q \kappa(\bm Q)^2} = \frac{1-\rho}{\tilde{q}} .
\end{align*}
Finally, the case that $\bm T$ is symmetric can be proven by the fact that $\bm Q$ is orthogonal, i.e., $\bm Q^{-1} = \bm Q^T$ and $\kappa(\bm Q)=1$.
Substituting this back into (\ref{equ:dTk}) and using the orthogonal invariance property of norm, we obtain the simplified version in (\ref{equ:dTk_sym}).
This completes our proof of Theorem~\ref{theo:vector}.

\end{appendices}


\bibliography{sn-bibliography}



\end{document}